\documentclass[a4paper,10pt]{article}
\usepackage{fullpage}
\usepackage{cite}
\usepackage{amsfonts}
\usepackage{textcomp}
\usepackage{xcolor}
\usepackage{graphicx}
\usepackage{color}
\usepackage{amsmath}
\usepackage{bbm}
\usepackage{amssymb}
\usepackage{amsthm}
\usepackage{xspace}
\usepackage{algorithm}
\usepackage{hyperref}
\usepackage{pdfsync}

\newif\ifshowcomments
\showcommentstrue
\ifshowcomments
\newcommand{\mynote}[2]{\fbox{\bfseries\sffamily\scriptsize{#1}}
 {\small$\blacktriangleright$\textsf{\emph{#2}}$\blacktriangleleft$}}
\else
\newcommand{\mynote}[2]{}
\fi

\usepackage[noend]{algpseudocode}
\algnewcommand{\LineComment}[1]{\State \(\triangleright\) #1}
\newtheorem{observation}{Observation}

\newtheorem{definition}{Definition}
\newtheorem{Lemma}{Lemma}
\newtheorem{theorem}{Theorem}

\def\BibTeX{{\rm B\kern-.05em{\sc i\kern-.025em b}\kern-.08em
    T\kern-.1667em\lower.7ex\hbox{E}\kern-.125emX}}
\begin{document}

\title{Byzantine Generalized Lattice Agreement
%\thanks{This work has been partially funded by...}
}

\author{Giuseppe Antonio Di Luna$^{*}$ , Emmanuelle Anceaume$^\dag$, Leonardo Querzoni$^{*}$ \\
\small
*: DIAG, Sapienza University, Roma, Italia.  \texttt{\{diluna,querzoni\}@diag.uniroma1.it}\\
\small
$\dag$: CNRS, Univ Rennes, Inria, IRISA, France. \texttt{emmanuelle.anceaume@irisa.fr}\\
}
%
%\author{Giuseppe Antonio Di Luna \footnote{ Diag, Sapienza University, g.a.diluna@gmail.com}, Emmanuelle Anceaume\footnote{CNRS, Univ Rennes, Inria, IRISA, emmanuelle.anceaume@irisa.fr}, Leonardo Querzoni  \footnote{ Diag, Sapienza University, querzoni@diag.uniroma1.it}}
%\input{preamble}
\date{}
\maketitle
\begin{abstract}

%This paper presents Wait Till Safe, a wait-free algorithm to solve the Lattice Agreement problem in asynchronous distributed systems where processes may exhibit Byzantine behaviors. It then introduces a Generalized version of the algorithm that allows processes to propose and decide an unbounded sequence of values. The paper further proposes the construction of a Replicated State Machine for objects with commuting update operations that guarantees both linearizability and progress in asynchronous environments. The main achievement of this construction is to tolerate Byzantine failures, which to the best of our knowledge has never been achieved before.

%\color{red}

The paper investigates the Lattice Agreement (LA) problem in asynchronous systems. In LA each process proposes an element $e$ from a predetermined lattice, and has to decide on an element $e'$ of the lattice such that $e \leq e'$.
Moreover, decisions of different processes have to be comparable (no two processes can decide two elements $e$ and $e'$ such that $(e \not\leq e') \land (e' \not\leq e)$).

It has been shown that Generalized LA (i.e., a version of LA proposing and deciding on sequences of values) can be used to build a Replicated State Machine (RSM) with commutative update operations.
The key advantage of LA and Generalized LA is that they can be solved in asynchronous systems prone to crash-failures (this is not the case with standard Consensus).  

In this paper we assume Byzantine failures.
We propose the Wait Till Safe (WTS) algorithm for LA, and we show that its resilience to $f \leq (n-1)/3$ Byzantines is  optimal.
We then generalize WTS obtaining a Generalized LA algorithm, namely GWTS. We use GWTS to build a RSM with commutative updates. Our RSM works in asynchronous systems and tolerates $f \leq (n-1)/3$ malicious entities.
All our algorithms use the minimal assumption of authenticated channels.
When the more powerful public signatures are available, we discuss how to improve the message complexity of our results (from quadratic to linear, when $f={\cal O}(1)$).
At the best of our knowledge this is the first paper proposing a solution for Byzantine LA that works on any possible lattice, and it is the first work proposing a Byzantine tolerant RSM built on it.  
\color{black}
\end{abstract}

\section{Introduction} % (fold)

State machine replication (RSM) is today the foundation of many cloud-based highly-available products: it allows some service to be deployed such to guarantee its correct functioning despite possible faults. In RSM, clients issue operation requests to a set of distributed processes implementing the replicated service, that, in turn, run a protocol to decide the order of execution of incoming operations and provide clients with outputs. Faults can be accidental (e.g. a computer crashing due to a loss of power) or have a malicious intent (e.g. a compromised server). Whichever is the chosen fault model, RSM has proven to be a reliable and effective solution for the deployment of dependable services. RSM is usually built on top of a distributed Consensus primitive that is used by processes to agree on the order of execution of requests concurrently issued by clients. The main problem with this approach is that Consensus is impossible to achieve deterministically in a distributed settings if the system is asynchronous and even just a single process may fail by crashing \cite{FLP}. This led the research community to study and develop alternative solutions based on the relaxation of some of the constraints, to allow agreement to be reached in partially synchronous systems with faulty processes by trading off consistency with availability.

An alternative approach consists in imposing constraints on the set of operations that can be issued by clients, i.e. imposing updates that commute. In particular, in 2012 Faleiro et al.~\cite{Faleiro:2012} introduced a RSM approach based on a generalized version of the well known Lattice Agreement (LA) problem, that restricts the set of allowed update operations to commuting  ones \cite{Shapiro:2018}. 
They have shown that commuting replicated data types (CRDTs) can be implemented with an RSM approach in asynchronous settings using the monotonic growth of a join semilattice, i.e., a partially ordered set that defines a join (least upper bound) for all element pairs (see Figure~\ref{slattice} for an example). A typical example is the implementation of a dependable counter with {\sf add} and {\sf read} operations, where updates (i.e. {\sf add}s) are commutative.

In the Lattice Agreement problem, introduced by Attiya  et al.~\cite{Attiya:1995}, each process $p_i$ has  an input value $x_i$ drawn from the join semilattice  and must decide an output value $y_i$, such that \emph{(i)} $y_i$ is the join of $x_i$  and some set of input values and \emph{(ii)} all output values are comparable to each other in the lattice, that is form a chain in the lattice (see Figure~\ref{slattice}). Lattice Agreement describes situations in which processes need to obtain some knowledge on the global execution of the system, for example a global photography of the system. In particular Attiya  et al.~\cite{Attiya:1995} have shown that in the asynchronous shared memory computational model, implementing a snapshot object is equivalent to solving the Lattice Agreement problem. 

Differently from Consensus, Lattice Agreement can be deterministically solved in an asynchronous setting in presence of crash failures. Faleiro  et al.~\cite{Faleiro:2012} have shown that a majority of correct processes and reliable communication channels are sufficient to solve Lattice Agreement, while Garg   et al.~\cite{Garg:2018} proposed a solution that requires $\mathcal O(\log n)$ message delays, where $n$ is the number of processes participating to the algorithm.  The very recent solution of Skrzypczak  et al.~\cite{Skrzypczak}  considerably  improves  Faleiro's construction in terms of memory consumption, at the expense of progress. 

In the Generalized Lattice Agreement (GLA) problem processes propose an infinite number of input values (drawn from an infinite semilattice) and decide an infinite sequence of output values, such that, all output values are comparable to each other in the lattice i.e. form a chain (as for Lattice Agreement); the sequence of decision values are non-decreasing, and every input values eventually appears in some decision values. Solving GLA in asynchronous distributed systems reveals to be very powerful as it allows to built a  linearizable RSM of commutative update operations~\cite{Faleiro:2012}. 

Despite recent advancements in this field, to the best of our knowledge no general solution exists that solves Lattice Agreement  problems in an asynchronous setting with Byzantine faults. In the present paper we continue the line of research on Lattice Agreement in asynchronous message-passing systems by considering a Byzantine fault model, that is a model where processes may exhibit arbitrary behaviors. 

%Clearly the challenge to be faced by a Generalized Lattice Agreement algorithm is to prevent that the receipt of an infinite number of input values at each process will lead processes in an endless execution: each time a process proposes a new input value, all  (correct) processes need to refine their current estimation of the decision value to guarantee that any two decision values are comparable by containment and are monotonically increasing. Since this refinement may be needed for each new value progress is not guaranteed. Faleiro  et al.~\cite{Faleiro:2012} circumvent this issue by relying on a batching mechanism: new values are processed only when the current batch of values has been decided. Their solution tolerates up to a minority of crash failures. 

%\subsection{Our Contribution}

We first introduce the LA specifications  that takes into account Byzantine faults. Then we propose an algorithm, namely \emph{Wait Till Safe} (WTS), which, in presence of less than $(n-1)/3$ Byzantine processes, guarantees that any correct process decides in no more than $5+2f$ message delays with a global message complexity in $\mathcal O(n^2)$ per process. We show that $(n-1)/3$ is an upper bound. Our algorithm makes use of a Byzantine reliable broadcast primitive to circumvent adversarial runs where a Byzantine process may induce correct processes to deliver different input  values.
%Its message complexity is in $O(n^2)$ due to the Byzantine reliable broadcast which  guarantees that for any message broadcast by Byzantine processes if one correct process receives it then all of them will receive it. This communication primitive combined with a ``wait until safe'' strategy ensures that adversarial scenarios in which Byzantine processes pretend to have a non comparable state  with correct processes is circumvented. 
The algorithm is wait-free, i.e., every process completes its execution of the algorithm within a bounded number of steps, regardless of the execution of other processes. Note that Nowak and Rybicki~\cite{Nowak:2019} recently proposed a solution for a more constrained specification of LA in a Byzantine setting where correct processes are not allowed to decide values that contain inputs from Byzantine processes; %\ea{they prove that solving this constrained  LA is impossible unless the total number of processes is larger than a function of both the number of faulty processes and the number of input values.} 
 Similarly to previous contributions in the area of Byzantine fault-tolerant agreement~\cite{kotla2007zyzzyva}, our specification allows that input values proposed by Byzantine processes may be included in some decisions from some correct processes; this choice allows us to sidestep the impossibility results of~\cite{kotla2007zyzzyva} and makes our solution suitable in important  application use cases where the set of input values is infinite (e.g. natural numbers).

We then go a step further by proposing an algorithm, namely \emph{Generalized Wait Till Safe} (GWTS), to solve GLA in a Byzantine fault model. Here the challenge is twofold: first, we need to guarantee that,  despite the fact that input values are proposed in tumbling batches, Byzantine processes cannot keep rejecting all new proposals under the pretext that they are not comparable with the current ones. Second, we must ensure that adversarial processes cannot progress much faster than all the other correct processes (i.e., output decision values faster than correct processes) which would allow them to prevent all correct processing from proposing their own values.
Our ``wait until safe'' strategy guarantees that each correct process performs an infinite sequence of decisions, and for each input received at a correct process, its value is eventually included in a decision. Our algorithm is wait-free and is resilient to $f \leq  (n-1)/3$  Byzantine processes.

We present the construction of a RSM for objects with commuting update operations that guarantees both linearizability and progress in asynchronous environments with Byzantine failures.%, which to the best of our knowledge has never been achieved before. 

Finally, we  sketch the main lines of a signature-based version of our algorithms which takes advantage of digital signatures to reduce the message complexity to $\mathcal O(n)$ per process, when $f = \mathcal O(1)$.

To the best of our knowledge this is the first paper proposing a solution for Byzantine LA that works on any possible lattice, and it is the first work proposing a Byzantine tolerant RSM built on it.

%\ea{Do we remove this to save space: To summarize:
%\begin{itemize}
%\item We formalize Lattice Agreement problems in a Byzantine setting;
%\item We propose a wait-free LA algorithm tolerant to $f$ Byzantine failures with $n \geq 3f+1$, in asynchronous environments;
%\item We propose a wait-free GLA algorithm  tolerant to $f$ Byzantine failures with $n \geq 3f+1$, in asynchronous environments that builds upon the previous one;
%\item We propose a construction of a RSM for conflict-free objects that is resilient to $f$ Byzantine failures with $n \geq 3f+1$, is wait-free and guarantees both linearizability and progress;
%\item Finally, we sketch a signature-based version of our solutions that improve performance with respect to message complexity.
%\end{itemize}}

The rest of this paper is organized as follows: Section \ref{sec:related_work} discussed the related works; Section \ref{sec:preliminary} illustrates the necessity of at least $(3f+1)$ processes, while Section \ref{sec:model} describes the system model; Section \ref{sec:wts} and Section \ref{sec:gwts} introduce the algorithms for Byzantine LA and Byzantine GLA, respectively; Section \ref{sec:byzantine_rsm} describe the construction of a byzantine tolerant RSM; Section \ref{sec:sign} sketches a signature-based variant of our solutions and, finally, Section \ref{sec:conclusions} concludes the paper.

\section{Related Work} % (fold)
\label{sec:related_work}

Lattice Agreement has been introduced by Attiya et al~\cite{Attiya:1995} to efficiently implement an atomic snapshot object~\cite{Anderson:1993,Attiya:1993}. Their construction is such that each \emph{scan} or \emph{collect} operation requires $\mathcal O(1)$ execution of LA and uses $\mathcal O(n)$ read/write registers. Then Faleiro~\cite{Faleiro:2012} have shown that GLA is a very interesting  abstraction to build RSMs with strong consistency properties, i.e., linearizability of its operations, and liveness guarantees in asynchronous systems. Very recently Nowak and Rybicki~\cite{Nowak:2019} have studied LA in presence of Byzantine faults. As previously introduced, their specifications of LA is more restrictive than the one we propose since it does not allow decisions to contain values proposed by Byzantine processes. 
We argue that our Lattice Agreement   specification is more adapted to build RSM on top of the LA algorithm. 
Removing a value proposed by a Byzantine process might not be desirable: think about an RSM that implements an object shared by different organizations, it could be a breach of contract to selectively 
avoid certain updates even when the sender misbehaved. 
A second reason is more technical and it stems from the interaction between the impossibility results introduced by the specifications of~\cite{Nowak:2019} and how an RSM is implemented using GLA~\cite{Faleiro:2012}.
Following~\cite{Faleiro:2012} to implement the RSM we take the power set of all possible updates and we construct a lattice on it using as join the union operation. As an example, let us suppose that we want to build a set counter data type,
and let us assume that clients issue four update operations {\sf add($1$)}, {\sf add($2$)}, {\sf add($3$)}, {\sf add($4$)} interleaved with reads. In this case our semilattice is the one in Figure \ref{slattice}, and a Lattice Agreement algorithm
will ensure that each read will see values on a single chain, the red one in the figure. Thanks to this, different reads will see ``growing" versions of the counter that are consistent snapshots (e.g., if someone reads $\{1\}$, the other could read $\{1,4\}$, but it can not read $\{4\}$). 
\begin{figure}
  \centering
    \includegraphics[width=0.5\textwidth]{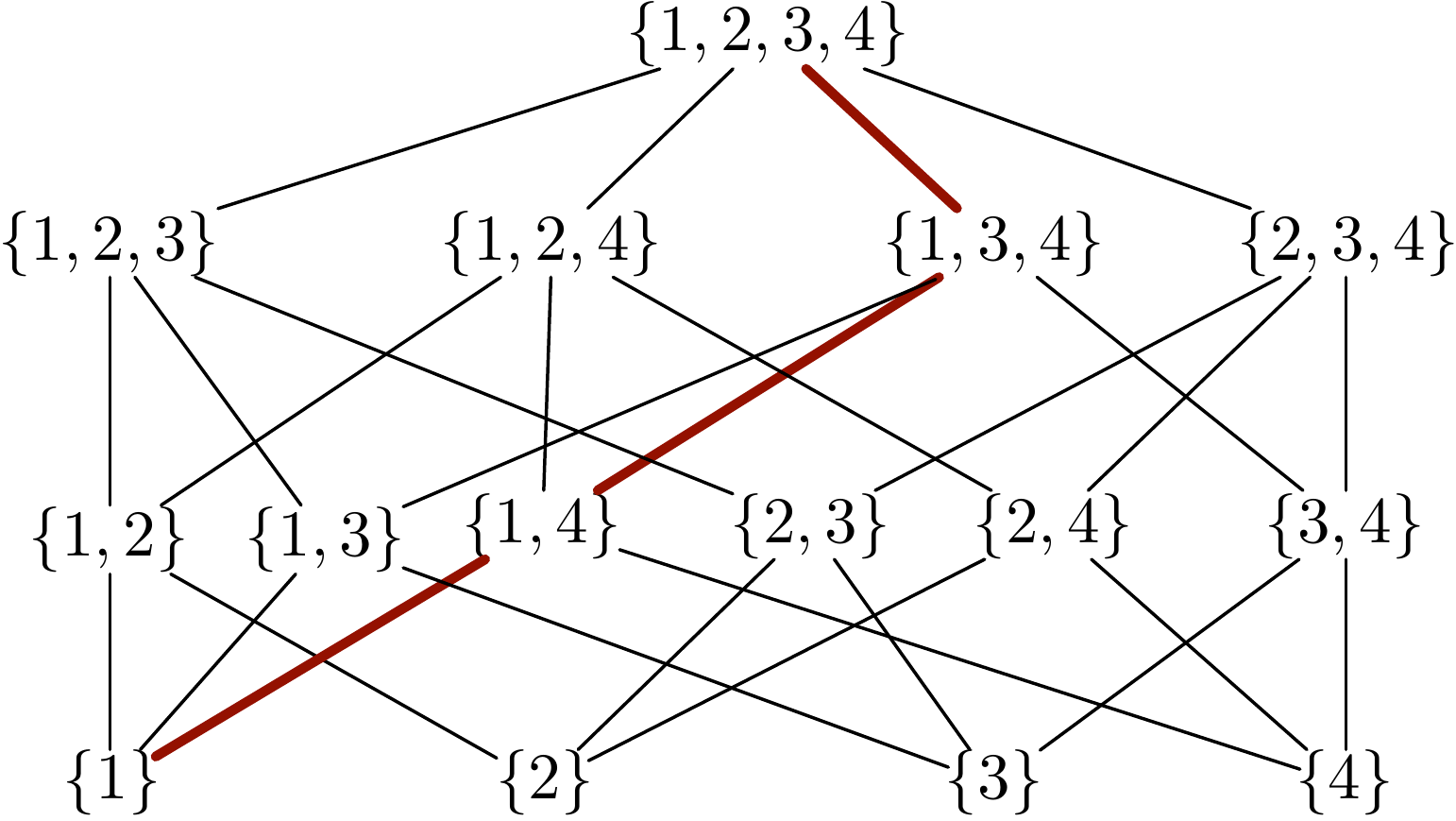}
      \caption{Hasse diagram of the semilattice induced over the power set of $\{1,2,3,4\}$ using the union operation as join. Taken two elements $e,e'$ of the lattice if $e < e'$, then $e$ appears lower in the diagram than $e'$ and there is an ``upward" path, going from lower points to upper points, connecting $e$ to $e'$ (e.g., $\{1\} \leq \{1,3,4\}$, but $\{2\} \not\leq \{3\}$). Any two elements $e,e'$ of the semilattice have a join $e \oplus e'= e \cup e'$ and $e \oplus e' \geq e,e'$ (e.g., $\{1\} \oplus \{2,3\}=\{1,2,3\}$ ). The red edges indicate the chain (i.e., sequence of increasing values) selected by the Lattice Agreement protocol.}\label{slattice}
\end{figure}
Such a semilattice has a breadth\footnote{Formally, the breadth of a semilattice $L=(V, \oplus)$ is the largest $n$, such taken any set of $U \subseteq V$ of size $n+1$ we have $\bigoplus U=\bigoplus K$ where $K$ is a proper subset of $U$.} of $4$, actually, each semilattice obtained using as join operation the union over the power set of a set of $k$ different values has breadth $k$. Therefore to satisfy
the specifications of \cite{Nowak:2019} using the semilattice in Figure \ref{slattice} we should have at least $5$ processes participating to Lattice Agreement. 
Unfortunately, it is often the case that the number of possible update operations is larger than the processes running the LA. In the set counter data type, we may have an {\sf add($x$)} for any $x \in \mathbb{N}$, in such setting the specifications of~\cite{Nowak:2019} is impossible to implement. 
Our specifications circumvent such impossibility.

\section{Model and Problem Definition} % (fold)
\label{sec:model}

We have a set $P:\{p_0,p_1,\ldots,p_{n-1}\}$ of processes. They communicate by exchanging messages over asynchronous authenticated reliable point-to-point communication links (messages are never lost on links, but delays are unbounded). 
The communication graph is complete: there is a communication link between each pair of processes. 

  We have a set $F \subset P$ of Byzantine processes, with $|F| \leq f$. Byzantine processes deviate arbitrarily from the algorithm.   We assume $|P| \geq 3f+1$ (we show that such assumption is needed in Section~\ref{sec:preliminary}).
  
  In Section~\ref{sec:sign}, we assume that there exists a public-key infrastructure, and that each process is able to sign a message, in such a way that each other process is able to unambiguously verify such signature. 
   In this case Byzantine processes are not able to forge a valid signature for 
  a process in $C=P \setminus F$ (that is the computational capability of processes in $F$ is polynomially bounded). We remark that signatures are only used in the algorithms presented in Section~\ref{sec:sign}.

\subsection{The Byzantine Lattice Agreement Problem}  
  Each process $p_i \in C$ starts with an initial input value $pro_i \in E \subseteq V$. 
Values in $V$ form a join semi-lattice $L=(V,\oplus)$ for some commutative join operation $\oplus$: for each $u,v \in V$ we have $u \leq v$ if and only if $v= u \oplus v$. The set $E$ is a set of lattice values that can be proposed by a process. 
 Given  $V'=\{v_1,v_2,\ldots,v_k\} \subseteq V$ we have $\bigoplus V'=v_1 \oplus v_2 \oplus \ldots \oplus v_k$.

  The task that processes in $C$ want to solve is the one of  Lattice Agreement, and it is formalised by the following properties:
  \begin{itemize}
\item {\bf Liveness:} Each process $p_i \in C$ eventually outputs a decision value $dec_i \in V$;
\item {\bf Stability:} Each process $p_i \in C$  outputs a unique decision value $dec_i \in V$;
\item {\bf Comparability:} Given any two pair $p_i,p_j \in C$ we have that either $dec_i \leq dec_j$ or $dec_j \leq dec_i$;
\item {\bf Inclusivity:} Given any correct process $p_i \in C$ we have that $pro_i \leq dec_i $;
\item {\bf Non-Triviality:}  Given any correct process $p_i \in C$ we have that $dec_i \leq \bigoplus (X \cup B)$, where $X$ is the set of proposed values of all correct processes ($X:\{pro_i | \text{ with } p_i \in C \}$), and $B \subseteq E$ is $|B| \leq f$.

  \end{itemize}
  
   In the rest of the paper we will assume that $L$ is a semi-lattice over sets ($V$ is a set of sets) and $\oplus$ is the set union operation. This is not restrictive, it is well known \cite{Nation} that any join semi-lattice is isomorphic to a semi-lattice of sets with  set union  as join operation.

\section{Necessity of at least $(3f+1)$ processes} % (fold)
\label{sec:preliminary}

We first show that our specification can only be satisfied when there are at least $(3f+1)$ processes. Interestingly, this bound holds even if we weaken our specification by removing Inclusivity property.
\begin{theorem}
Let ${\cal A}$ be any asynchronous distributed, and Byzantine tolerant algorithm that solves Byzantine Lattice Agreement when $f$ processes are Byzantine. 
We have that ${\cal A}$ needs at least $(3f+1)$ processes. This holds even if we drop the Inclusivity property from the specification.
\end{theorem}

\begin{proof}({\em Sketch})
We first show the result for $f=1$. Let ${\cal A}$ be a correct algorithm using $3$ processes and tolerating a Byzantine process. 
The processes are $\{p_1,p_2,p_{biz}\}$, and $p_{biz}$ is the Byzantine.
Let us consider a run where we delay the messages between $p_1$ and $p_2$. 
Still ${\cal A}$  has to make $p_1$ and $p_2$ decide before they exchange any message: each one of them could not exclude that the other is the Byzantine process. 
Suppose also that the byzantine pretend to have proposed value $v_{biz_1}$ while talking with $p_1$, and $v_{biz_2}$ while talking with $p_2$.
Therefore $p_1$ has to make a decision seeing only its value $v_1$ and $v_{biz_1}$. 
Also $p_2$ sees only values  $v_2$ and $v_{biz_2}$. 
Now we have the following possibilities:
\begin{itemize}
\item Each of them outputs a decision that contains both values, but then decisions are not comparable.
\item Each of them outputs a decision that contains only its proposed value, but then the decisions are not comparable.
\item Each of them outputs a decision that contains only the Byzantine value, also in this case the decision are not comparable.
\item One of them, let us say $p_1$ outputs $v_1$, and the other outputs $v_2,v_{biz_2}$. It is clear that also this case does not solve the problem. 
\end{itemize}
For $f>1$ we use the standard simulation argument: given an algorithm that solves the problem with $3f$ processes we could simulate it using only $3$ processes (each of the three
internally runs $f$ copies). 
\end{proof}

\section{Algorithm {\em  Wait Till Safe} (WTS)} % (fold)
\label{sec:wts}

%The main idea is to make any process to disclose its proposal by performing a reliable broadcast. The reliable broadcast prevents Byzantine processes from sending different messages to processes~\cite{Bracha:1987,Srikanth:1987}. The exact specification of this broadcast are in~\cite{gargdistcomp}.

The \emph{Wait Till Safe} algorithm (Algorithms \ref{wts:algorithm} and \ref{wts:acceptoralgorithm}) is divided in two phases: an initial \emph{Values Disclosure Phase} where processes are asked to declare to the whole system values they intend to propose, and then a \emph{Deciding Phase} where processes agree on which elements of the lattice can be decided on the basis of the proposed values. For the sake of clarity, processes are divided in \emph{proposers} that propose an initial value, and then decide one decision value, and \emph{acceptors} which help proposers decide. This distinction does not need to be enforced during deployment as each process can play both roles at the same time.

The main idea in the  \emph{Values Disclosure Phase} is to make any proposer disclose its proposed value by performing a Byzantine reliable broadcast. The reliable broadcast prevents Byzantine processes from sending different messages to processes~\cite{Bracha:1987,Srikanth:1987}. The exact specification of this broadcast is in~\cite{gargdistcomp}. In the pseudocode the broadcast primitive is represented by the {\sc ReliableBrodcast} (used for reliably broadcast messages) function and {\sc RBcastDelivery} event (that indicates the delivery of a message sent with the reliable broadcast).

%We report it in the following for completeness. Each process broadcasts a message $m$ that its decorated with an epoch number $ep$ and its identity $p_i$.
% \begin{itemize}
% \item {\bf Integrity:} If a correct process $p_i$ never broadcasts $<p_i, m, ep>$, no other correct process will ever deliver $<p_i, m, ep>$;
% \item {\bf Liveness:} If a correct process $p_i$ broadcasts $<p_i, m, ep>$, all other correct  processes eventually deliver $<p_i, m, ep>$;
% \item {\bf Global liveness:} For any two correct  processes $p_i$ and $p_j$, if $p_i$ delivers $<p_k, m, ep>$, $p_j$ also delivers
% $<p_k, m, ep>$, even when the sender $p_k$  is Byzantine.
% \item {\bf Global uniqueness:} If a correct process $p_i$ delivers $<p_s,m,ep>$ then for any other correct process $p_j$ that delivers $<p_s,m',ep>$  we have $m=m'$ irrespectively of the correctness of $p_s$.
%
% \end{itemize}

Values delivered at each process are saved in a $SvS$ (Safe-values Set). A process moves to the next phase as soon as he receives values from at least $(n-f)$ proposers. Waiting for $(n-f)$ messages is not strictly necessary, but allows us to show a bound of ${\cal O}(f)$ on the message delays of our algorithm. Note that, from this point on, some operations of Phase 1 could run in parallel with Phase 2. 
Thanks to \emph{Value Disclosure Phase} a process is {\em committed} to its value and cannot change its proposal or introduce a new one during the \emph{Deciding Phase}. During this latter phase, correct processes only handle messages that contain values in $SvS$, i.e. messages for which the $SAFE()$ predicate is true. Messages that do not satisfy this condition are buffered for later use, i.e. if and when all the values they contain will be in $SvS$.

The \emph{Deciding Phase} is an extension of the algorithm described in~\cite{Faleiro:2012} with a Byzantine quorum and additional checks used to thwart Byzantine attacks.
Each correct proposer $p$ sends its $Proposed$ value to acceptors in a request message (Line \ref{firstproposal}). An acceptor receiving a request, sends an ack if the previously $Accepted\_set$ is a subset of the value contained in the request, and updates its $Accepted\_set$ using the $Proposed$ set in the request (initially, the $Accepted\_set$ of an acceptor is the empty set). 
Otherwise, the acceptor sends a nack containing the $Accepted\_set$, and it updates its $Accepted\_set$ with the union of the value contained in  $p$'s request and its $Accepted\_set$. 
The proposer $p$ decides if it receives $ \lfloor (n+f)/2 \rfloor +1$ acks. In case $p$ receives a nack, then $p$ updates $Proposed$ set by taking the union of it and the value contained in the nack. 
Each time a proposer updates its $Proposed\_set$ set it issues a new request.

\begin{algorithm*}
\caption{{\em WTS (Wait Till Safe)}  -Algorithm for Proposer process $p_i$ }\label{wts:algorithm}
\footnotesize

\begin{algorithmic}[1]
\State $V$ lattice values, $E$ subset of $V$ of values that can be proposed by processes. 
\State $proposed\_value=pro_i$
\State $init\_counter=ts=0$
\State $Proposed\_set=Ack\_set=SvS=Waiting\_msgs=\emptyset$
\State $state = disclosing$

\smallskip

\algrenewcommand\algorithmicprocedure{\textbf{upon event}}

\LineComment{Values Disclosure Phase}
\Procedure{$proposed\_value \neq \bot$}{}
\State  $Proposed\_set=Proposed\_set \cup proposed\_value$  \label{insertmine}
\State  {\sc ReliableBroadcast}$(<disclosure\_phase, proposed\_value>)$ to all  \label{rb1}
\EndProcedure
\smallskip

%\Comment{A process rejects any message in which value is not well-formed (i.e. is not a point of the lattice)}
\Procedure{RBcastDelivery from sender}{ $<disclosure\_phase,value>$}
\If{$value \in E$}\label{wts:checke}
\If{$state = disclosing$}
\State  $Proposed\_set=Proposed\_set \cup value$ 
\EndIf
\State $SvS = SvS \cup  \{value\}$\label{wts:updatesvs}
\State $init\_counter=init\_counter+1$
\EndIf

\EndProcedure
\smallskip

\LineComment{Deciding Phase}
\Procedure{$init\_counter \geq (n-f)$ when $state = disclosing$ }{}
\State $state = proposing$
\State  {\sc Broadcast}$(<ack\_req, Proposed\_set,ts>)$ to all Acceptors \label{firstproposal}
\EndProcedure
\smallskip

\Procedure{Delivery from sender}{$m$}
\State  $Waiting\_msgs=Waiting\_msgs \cup \{m\}$
\EndProcedure
\smallskip

%\Procedure{$\exists m \in Waiting\_msgs$ such that all values in $m$ are in $SvS$}{}
%\State  $Waiting\_msgs=Waiting\_msgs \setminus \{m\}$
%\State route message $m$ to the appropriate internal delivery function. 
%\EndProcedure
%\smallskip

\Procedure{ $\exists m \in Waiting\_msgs \,| \, SAFE(m) \land state=proposing \land m = <ack, \cdot, tstamp> \land tstamp=ts$ from sender}{}
\State  $Waiting\_msgs=Waiting\_msgs \setminus \{m\}$
\State  $Ack\_set=Ack\_set \cup \{<ack,sender>\}$
\EndProcedure
\smallskip

\Procedure{ $\exists m \in Waiting\_msgs \,| \, SAFE(m) \land state=proposing \land m = <nack,Rcvd\_set, tstamp> \land tstamp=ts$ from sender}{}
\State  $Waiting\_msgs=Waiting\_msgs \setminus \{m\}$
\If{$Rcvd\_set \cup Proposed\_set \neq Proposed\_set$} \label{ifref1}
\State $Proposed\_set=Rcvd\_set \cup Proposed\_set$
\State $Ack\_set=\emptyset$
\State $ts=ts+1$ %\leoq{update with $ts$}
\State  {\sc Broadcast}$(<ack\_req, Proposed\_set,ts>)$ to all Acceptors \label{pref1}
\EndIf
\EndProcedure
\smallskip

\Procedure{$|Ack\_set| \geq  \lfloor (n+f)/2 \rfloor +1  $ when $state=proposing$}{}
\State $state=decided$ \label{oneshot}
\State $decision_i=Proposed\_set$
\State {\sc DECIDE}($decision_i$)
\EndProcedure
\smallskip

\algrenewcommand\algorithmicprocedure{\textbf{function}}

\Procedure{SAFE}{$m$}  \label{w1}
\If{the lattice element contained in $m$ is a subset of $SvS$}
\State {\bf return} True
\Else
\State {\bf return} False
\EndIf  \label{w2}

\EndProcedure

\end{algorithmic}
\end{algorithm*}

\begin{algorithm*}
\caption{{\em WTS (Wait Till Safe)}  - Algorithm for Acceptor process $p_i$} \label{wts:acceptoralgorithm}
\footnotesize

\begin{algorithmic}[1]
\algrenewcommand\algorithmicprocedure{\textbf{upon event}}
\State $Accepted\_set = Waiting\_msgs=\emptyset$
\State $SvS$ \Comment{Reference to SvS in the corresponding Proposer}

\Procedure{Delivery from sender}{$m$}
\State  $Waiting\_msgs=Waiting\_msgs \cup \{m\}$
\EndProcedure
\smallskip

\Procedure{ $\exists m \in Waiting\_msgs \,| \, SAFE(m) \land  m = <ack\_req,Rcvd\_set,x>$ from sender}{}
\State  $Waiting\_msgs=Waiting\_msgs \setminus \{m\}$
\If{ $Accepted\_set \subseteq
 Rcvd\_set$} \label{ifack}
\State  $Accepted\_set=Rcvd\_set$ \label{awinsert}
\State {\sc Send to sender} $(<ack, Accepted\_set,x>)$
\Else
\State  {\sc Send to sender} $(<nack, Accepted\_set,x>)$
\State $Accepted\_set=Accepted\_set \cup Rcvd\_set$
\EndIf
\EndProcedure

\end{algorithmic}
\end{algorithm*}

\subsection{WTS properties}

\subsubsection{Safety properties}\label{safetyoneshot}

\begin{observation} Given any correct process $p_j$ its $SvS$ contains at most one value for each process in $P$.  \label{obssvs}
\end{observation}

The above observation derives from the specification of reliable broadcast, and the fact that in the disclosure phase each participating process broadcasts a single value. 
We say that a message $m$ containing a set of proposed values is ``safe" for a process $p_i$ if such set of values is contained in $SvS$. 
It is immediate from function at Lines~\ref{w1}-\ref{w2} that proposers (in state $proposing$)  change their $Proposed\_set$ only
when they receive safe messages.
The analogous holds for the $Accepted\_set$ of acceptors.

We say that a value $v$ receives $m$ acks if it is contained in a $Proposed\_set$, that is in turn contained in ack messages in the form $<ack,\cdot,ts>$ sent by $m$ acceptors.
The same meaning is intended when we say that $Proposed\_set$ receives acks.

\begin{definition} (Commited value) A value $v$ is commited if it received $ \lfloor (n+f)/2 \rfloor +1$ acks. 
\end{definition}
\begin{definition} (Commited proposal) A $Proposed\_set$ is commited if it received $\lfloor (n+f)/2\rfloor +1$ acks. \label{def:commitedproposal}
\end{definition}

\begin{Lemma}\label{Lemmacomm}
Let $t$ be the first time at which a value $v$ is commited, we have that any  $Proposed\_set$ committed after $t$ contains $v$. 
\end{Lemma}
\begin{proof}
Value $v$ received at least $\lfloor (n-f)/2\rfloor +1$ acks from acceptors in $C$ (see Alg.~\ref{wts:algorithm} Line~\ref{livenessRSM-b}). These acceptors have inserted $v$ in their $Accepted\_set$ (see  Alg.~\ref{wts:acceptoralgorithm} Line~\ref{awinsert}). 
Thus by time $t$ a set $Q_{1} \subseteq C$ has $v$ in their  $Accepted\_set$.
Let $Proposed\_set$ be a value committed after $t$, then, by the same above reasoning, we have that $Proposed\_set$ received acks from a set  of correct acceptors $Q_2 \subseteq C$, with $|Q_2| \geq \lfloor (n-f)/2\rfloor +1$. 
Since $\exists p \in Q_1 \cap Q_2$  (recall that $|C|=(n-f)$) we have that $Proposed\_set$ contains $v$: $p$ sent an ack, thus has passed the if at Line~\ref{ifack} of Alg~\ref{wts:acceptoralgorithm}.
\end{proof}

\begin{observation} Given any correct process $p_j$ its $decision_j$ has been committed. \label{obsobv}
\end{observation}

\begin{theorem}
Let us consider a set of processes, of size at least $(3f+1)$, executing WTS algorithm. 
Algorithm WTS enforces: (1) {Comparability}; (2) { Inclusivity}; (3) { Non-Triviality}; (4) { Stability}. 
\end{theorem}
\begin{proof}
We prove each property separately. 
\begin{enumerate}
 \item is implied by Lemma~\ref{Lemmacomm} and Observation~\ref{obsobv}.
 \item derives from the fact that a proposer never removes a value from $Proposed\_set$ and from Line~\ref{insertmine}. 
\item the bound on $B$ derives from the safety of messages and Observation~\ref{obssvs}, the fact that $dec_i \leq \bigoplus (X \cup B)$ derives from the fact that a correct process insert in its proposal 
only values received by messages and its initial proposed value. The fact that $B \subseteq E$ derives from Line \ref{wts:checke} of Algorithm \ref{wts:algorithm}.
\item is ensured by Line~\ref{oneshot} in proposers.  
\end{enumerate}\end{proof}

Note that the {Inclusivity} and the {Comparability} imply that, when all correct proposers decide than each value proposed by some correct will be in a decision and that there exists a proposer
whose decision includes all values proposed by correct proposers. 

\subsubsection{Liveness properties}

\begin{Lemma}\label{Lemmasafe}
Each message sent by a correct process is eventually safe for any other correct process.
\end{Lemma}
\begin{proof}
If a correct process $p_i$ sent a message $m$ then the set of values contained in $m$ is a subset of $SvS$ of $p_i$. 
Note that $SvS$ is only updated as result of the reception of a message reliably broadcast in the disclosure phase (Line~\ref{wts:updatesvs}). From the properties of the broadcast eventually each other correct
process will obtain a $SvS$ that contains the set of values in $m$, making $m$ safe. \end{proof}

\begin{Lemma}\label{Lemmaref}
A correct proposer refines its proposal (executing Line~\ref{pref1}) at most $f$ times. 
\end{Lemma}
\begin{proof}
Each time the proposer executes Line~\ref{pref1} it passes the if at Line~\ref{ifref1}, thus increasing its proposed set. However, its first proposal, in Line~\ref{firstproposal}, contains at least $|X \cup B|-f$ values. Since
there are at most $|X \cup B|$ safe values (from Oservation~\ref{obssvs}), the claim follows.
\end{proof}

\begin{Lemma}\label{Lemmaeve}
If there is a time $t$ after which a correct proposer $p_i$ in state $proposing$  cannot execute Line~\ref{pref1}, then $p_i$ eventually decides. 
\end{Lemma}
\begin{proof}
Let $<ack\_req,Proposed\_set,ts>$ be the last ack request message sent by $p_i$. Since $p_i$ does not execute Line~\ref{pref1} it means that either $p_i$ does not receive any nack, or that any nack $p_i$ receives 
does not allow him to pass the if Line~\ref{ifref1}. Since $p_i$ is correct its message  $<ack\_req,Proposed\_set,ts>$ will reach each correct acceptor. By hypothesis each of them will send a ack, otherwise $p_i$ should
be able to execute Line~\ref{pref1} (they all handle the ack request by Lemma~\ref{Lemmasafe}). Once $p_i$ receives the acks from the set of correct acceptors will handle them: these messages are safe since
the element of the lattice inside each of them is equal to $Proposed\_set$ in $ack\_req$ and, thus $p_i$ decides. 
\end{proof}
In the next Theorem (Th.~\ref{liveness:wts}) we will show that each correct process eventually commits and decides, we also bound the  number of delays needed, by each correct proposer, to reach a decision.

\begin{theorem}\label{liveness:wts}
Let us consider a set of processes, of size at least $(3f+1)$, executing WTS algorithm. 
Every correct proposer decides in at most $2f+5$ message delays. 
\end{theorem}
\begin{proof}
The reliable broadcast at Line~\ref{rb1} takes at most $3$ message delays. Therefore after three rounds each correct process start its first proposal. Each refinement takes at most $2$ message delay, the time needed to broadcast
and receive a response. There are $f$ refinements, see Lemma~\ref{Lemmaref}, executed in at most $2f+2$ message delays, and thus by Lemma~\ref{Lemmaeve} after $2f+5$ message delays a correct decides.
\end{proof}

Note that Theorem~\ref{liveness:wts} implies the {Liveness} property of our Lattice Agreement specification.

\subsubsection{Message complexity}
The Byzantine reliable broadcast used at Line~\ref{rb1} costs ${\cal O}(n^2)$ messages \cite{gargdistcomp}, this cost dominates the other algorithm operations: in the  $2f+5$ delays needed to reach the decision at most ${\cal O}(f \cdot n)$ messages are generated. 

\begin{algorithm*}
\caption{{\em GWTS }  -Algorithm for proposer process $p_i$ }\label{gwts:prop}
\footnotesize

\begin{algorithmic}[1]
\State $V$ lattice values, $E$ subset of $V$ of values that can be proposed by processes. 
\State $proposed\_value=pro_i$
\State $Batch[\forall r \in \mathbb{N}]=SvS[\forall k \in \mathbb{N}]=\emptyset$ \Comment{Array of value sets, one batch for each round}
\State $Counter[\forall r \in \mathbb{N}]=0$ \Comment{Array of numbers, one for each round}
\State $r=-1$
\State $ts=0$
\State  $Proposed\_set =Decided\_set=Waiting\_msgs=Ack\_history=\emptyset$
\State $state = newround$

\smallskip

\algrenewcommand\algorithmicprocedure{\textbf{upon event}}

\Procedure{new value}{v}
\State $Batch[r+1]=Batch[r+1] \cup \{v\}$
\EndProcedure
\smallskip

\LineComment{Values Disclosure Phase}

\Procedure{$state = newround$}{}
\State $state = disclosing$
\State $r=r+1$
\State  $Proposed\_set=Proposed\_set \cup Batch[r]$
\State  {\sc ReliableBroadcast}$(<disclosure\_phase, Batch[r],r>)$ to all \label{gwt:rbcast1}
\EndProcedure
\smallskip

\Procedure{RBcastDelivery from sender}{$<disclosure\_phase, Set,round>$}
\If{$\forall e \in Set,\,\, e \in E$}\label{gwt:propo}
\If{$state = disclosing$} \label{gwt:lattice}
\State  $Proposed\_set=Proposed\_set \cup Set$ \label{gwts:cumulation}
\EndIf
\State $SvS[round] = SvS[round] \cup Set$
\State $Counter[round]=Counter[round]+1$
\EndIf
\EndProcedure

\smallskip

\LineComment{Deciding Phase}
\Procedure{$Counter[r] \geq (n-f)$ when  $state = disclosing$}{} \label{gwt:label}
\State $state = proposing$
\State $ts=ts+1$
\State  {\sc Broadcast}$(<ack\_req, Proposed\_set,ts ,r>)$ to all Acceptors \label{gwt:proposefirst}
\EndProcedure

\Procedure{Delivery or RBcastDelivery from sender}{$m$}
\State  $Waiting\_msgs=Waiting\_msgs \cup \{m\}$
\EndProcedure
\smallskip

\Procedure{$\exists m \in Waiting\_msgs | SAFE(m) \land state=proposing \land m=<nack,Rcvd\_set, ts', r'> \land ts'=ts \land r'=r$}{}
\State  $Waiting\_msgs=Waiting\_msgs \setminus \{m\}$
\If{$Rcvd\_set \cup Proposed\_set \neq Proposed\_set$}
\State $Proposed\_set=Rcvd\_set \cup Proposed\_set$ \label{gwtsrefined}
\State $ts=ts+1$
\State  {\sc Broadcast}$(<ack\_req, Proposed\_set,ts,r>)$ to all Acceptors \label{gwt:refine}
\EndIf
\EndProcedure

\smallskip
%
%\LineComment{Messages with a round or a timestamp different from current $r,ts$ are ignored}
%\Procedure{ InternalP2PDelivery {\bf when} $state=proposing$ from sender}{$<nack,Rcvd\_set, ts, r>$}
%\If{$Rcvd\_set \cup Proposed\_set \neq Proposed\_set$}
%\State $Proposed\_set=Rcvd\_set \cup Proposed\_set$ \label{gwtsrefined}
%\State $ts=ts+1$
%\State  {\sc Brodcast}$(<ack\_req, Proposed\_set,r>)$ to all Acceptors
%\EndIf
%\EndProcedure
%\smallskip

\Procedure{$\exists m \in Waiting\_msgs | SAFE(m) \land state=proposing \land m=<ack, Accepted\_set, destination,sender ,timestamp,round> \land m$ was delivered with {\sc RBcastDelivery}}{}
\State  $Waiting\_msgs=Waiting\_msgs \setminus \{m\}$
\State $Ack\_history=Ack\_history \cup \{<ack, Accepted\_set, destination, sender,timestamp,round,>\}$
\EndProcedure
\smallskip

\Procedure{ $<ack, Accepted\_set, destination, \cdot ,timestamp,round>$ appears $\lfloor (n+f)/2 \rfloor +1$ times in $Ack\_history$}{}\label{livenessRSM-b}
\If{$Decided\_set \subseteq Accepted\_set$ and  $state = proposing \land round=r$} \label{gwtstability}
\State {\sc Decide}$(Accepted\_set)$\label{livenessRSM-e}
\State $Decided\_set=Accepted\_set$
\State $state = newround$ \label{abacus}
\EndIf
\EndProcedure

\Procedure{SAFE}{$m$}  \label{gw1}
\If{the lattice element contained in $m$ is a subset of $SvS[r]$}
\State {\bf return} True
\Else
\State {\bf return} False
\EndIf  \label{gw2}
\EndProcedure

\end{algorithmic}
\end{algorithm*}

\begin{algorithm*}
\caption{{\em GWTS}  - Algorithm for Acceptor process $p_i$}\label{gtws:acceptor}
\footnotesize

\begin{algorithmic}[1]
\algrenewcommand\algorithmicprocedure{\textbf{upon event}}
\State $Accepted\_set = Waiting\_msgs=Ack\_history=\emptyset$ 
\State $SvS[]$ \Comment{Reference to SvS in the corresponding Proposer} 
\State $Safe\_r=0$  \Comment{Max round for which it is safe to process messages}
\smallskip
\Procedure{Delivery or RBCastDelivery from sender}{$m$}
\State  $Waiting\_msgs=Waiting\_msgs \cup \{m\}$
\EndProcedure
\smallskip

%\Procedure{$\exists m \in Waiting\_msgs$ with round $r \leq Safe\_r$  such that all values in $m$ are in $SvS[r]$}{}
%\State  $Waiting\_msgs=Waiting\_msgs \setminus \{m\}$
%\State route message $m$ to the appropriate internal delivery function. 
%\EndProcedure
%\smallskip

\smallskip
\Procedure{$\exists m \in Waiting\_msgs | SAFEA(m) \land r \leq Safe\_r \land m=<ack\_req,Rcvd\_set,ts,r>$}{}
\State  $Waiting\_msgs=Waiting\_msgs \setminus \{m\}$
\If{ $Accepted\_set \leq
 Rcvd\_set$}
\State  $Accepted\_set=Rcvd\_set$ \label{gwtaa}
\State {\sc ReliableBroadcast}$(<ack, Accepted\_set, sender,p_i,ts,r>)$ to all \label{rback}
\Else
\State  {\sc Send to sender} $(<nack, Accepted\_set,ts,r>)$
\State $Accepted\_set=Accepted\_set \cup Rcvd\_set$ \label{gwtab}
\EndIf
\EndProcedure

\smallskip
\Procedure{$\exists m \in Waiting\_msgs | SAFEA(m) \land r \leq Safe\_r \land m=<ack, Accepted\_set, destination,sender ,ts,r> \land m$ was delivered with {\sc RBcastDelivery}}{}
\State  $Waiting\_msgs=Waiting\_msgs \setminus \{m\}$
\State $Ack\_history=Ack\_history \cup \{<ack, Accepted\_set, destination, sender,ts,r,>\}$
\EndProcedure

\smallskip
\Procedure{ $<ack, Accepted\_set, destination, \cdot ,ts,r>$ appears $\lfloor (n+f)/2 \rfloor +1$ times in $Ack\_history$}{} \label{gwts:accround}
\If{$r = Safe\_r$}
\State $Safe\_r=Safe\_r+1$ 
\EndIf
\EndProcedure

\Procedure{SAFEA}{$m$}  \label{gaw1}
\If{$\exists r$ such that the lattice element contained in $m$ is a subset of $SvS[r]$}
\State {\bf return} True
\Else
\State {\bf return} False
\EndIf  \label{gaw2}
\EndProcedure
\end{algorithmic}
\end{algorithm*}
\section{  Algorithm {\em Generalized Wait Till Safe} (GWTS)}% (fold)
\label{sec:gwts}

\subsection{The Generalised Byzantine Lattice Agreement Problem}
In the generalised version of our problem, each process $p_i$  receives, asynchronously, input values from an infinite 
sequence $Pro_i=\langle pro_{0},pro_{1}, pro_{2},\ldots \rangle$. Each $pro_i$ belongs to a set $E$ of values that can be proposed, note that $E$ is not necessarily finite. 

  \begin{itemize}
\item {\bf Liveness:} each correct process $p_i \in C$ performs an infinite sequence of decisions $Dec_{i}= \langle dec_{0},dec_{1}, dec_{2},\ldots \rangle$;
\item {\bf Local Stability:} For each $p_i \in C$  its sequence of decisions is non decreasing (i.e., $dec_{h} \subseteq dec_{h+1}$, for any $dec_{h} \in Dec_{i}$);
\item {\bf Comparability:} Any two decisions of correct processes are comparable, even when they happen on different processes; 
\item {\bf Inclusivity:} Given any correct process $p_i \in C$, if $Pro_i$ contains a value $pro_k$, then $pro_k$ is eventually included in $dec_{h} \in Dec_{i}$;
\item {\bf Non-Triviality:}  Given any correct process $p_i \in C$ if $p_i$ outputs some decision $dec_k$ at time $t$, then 
 \\ $dec_k \leq \bigoplus (Prop[0:h] \cup B[0:b])$. Where:

 $Prop[0:h]$ is the union of the prefixes, until index $h$, of all sequences $Pro_i$ of correct processes; and, $B[0:b]$ is the union of all prefixes, until index $b$,  of 
 $f$  infinite sequences $B_i$, one for each Byzantine process.  Each $B_i$ is a sequence of values in $E$.

  \end{itemize}

Intuitively, with {Non-Triviality}  we are bounding the number of values that the Byzantine processes could insert in any decision to a finite number of values. 

\subsection{Algorithm Description}
The pseudocode of GWTS is in Algorithms~\ref{gwts:prop}-\ref{gtws:acceptor}.

The \emph{Generalized Wait Till Safe} algorithm is an extension of the WTS algorithm based on the same batching approach proposed in~\cite{Faleiro:2012}. Input values at proposers are batched until a new decision round starts. Each decision round follows the two-phases approach of WTS. Note that rounds are executed asynchronously at each proposer.\footnote{The Byzantine reliable broadcast primitive used in \cite{gargdistcomp} is designed to avoid possible confusion of messages in round based algorithms. That is exactly what we need.}

Compared to WTS, an additional challenge to be faced is to prevent adversarial processes from indefinitely postpone the decision correct processes. A uncareful design could allow byzantine proposers to continuously pretend to have decided, thus jumping to new rounds, and clogging the proposers with a continuous stream of new values. This would make acceptors to continuously nack proposals of correct processes. We solve this problem through the acceptors. Acceptors will help a new proposal to be decided in round $r\geq 1$ when, and if, in round $(r-1)$   a proposal has been accepted by at least a (Byzantine) quorum of acceptors (i.e., $safe\_r = r$ ). 
In order for this to work we make acceptors to reliably broadcast their ack messages, in this way the acceptance of proposals is made public. Any correct proposer can decide, in a round $r$, any proposal that has been correctly accepted in round $r$, even if it was not proposed by itself (provided that such decision preserves the Local Stability). 
\subsection{GWTS properties}

\subsubsection{Safety properties}

The proof of the safety properties of GWTS is analogous to the proof contained in Section~\ref{safetyoneshot}. 
From the properties of reliable broadcast we have the following:
\begin{observation}\label{finitevalues}
For each correct process $p_i$ and each round $r$, the set $SvS[r]$ contains at most $n$ sets. 
\end{observation}

\begin{theorem}\label{th:gwtssafety}
Let us consider a set of processes, of size $(3f+1)$, executing GWTS algorithm. 
Algorithm GWTS enforces: (1) { Comparability}; (2) { Non-Triviality}; (3) { Stability}.
\end{theorem}
\begin{proof}
We prove each property separately.
\begin{enumerate}
 \item is implied by Lemma~\ref{Lemmacomm} and Observation~\ref{obsobv}.

\item For the Non-Triviality we have to show that for each correct proposer $p_i$ and each $dec_k \in Dec_{i}$ it holds  $dec_t \leq \bigoplus ( Prop[0:h] \cup B[0:b])$. First notice that $dec_t$ may only contain values that are present in $\bigcup_{r' \in [0,r]} SvS[r']$ with $r$ rounds in which $dec_k$ happens.
This derives from the fact that a correct process, at a given round, only handles messages that contain safe values. Let  $W_{r}=\bigcup_{r' \in [0,r]}SvS[r']$, it is immediate to see that $W_{r}$ contains at most the union of all values in the prefixes $Prop_i[0:h]$ from some index $h$, which correspond precisely to all values proposed by correct
processes until round $r$. It is also immediate that  $W_{r}$ contains at most the union of all values that Byzantine proposers have reliably broadcast in the first $r$ disclosure phases: this is equivalent to say that it contains a prefix of all the infinite sequences of values that Byzantine cumulatively
broadcast in the disclosure phases of our algorithm.  From these arguments, and the check at Line \ref{gwt:propo} of Algorithm \ref{gwts:prop}, the Non-Triviality follows. 
\item is ensured by line~\ref{gwtstability} in proposers. \end{enumerate}\end{proof}

\subsubsection{Liveness}

We say that a correct process $p_i$ joins a round $r$ if it sends a message $<disclosure\_phase, \cdot,r>$. 
Similarly a correct process $p_i$ proposes a $Set$ at round $r$ if it sends a message  $<ack\_req, Set,\cdot,r>$.

We say that a message $m$ is safe for a process $p_i$ at round $r$, if $SvS[r]$ of $p_i$ contains of all values contained in $m$. 

\begin{Lemma}\label{gwtLemmasafe}
Each message sent by a correct process at round $r$ is eventually safe, at round $r$, for any other correct process.
\end{Lemma}
\begin{proof}
Same of Lemma~\ref{Lemmasafe}
\end{proof}

\begin{definition}
We say that a round $r$ has a {``legitimate end"} if there exists a proposal that has been committed at round $r$.
\end{definition}
(See Definition~\ref{def:commitedproposal} of committed proposal) 
\begin{definition} Round $r$ is a legit round, at time $t$, if, either,  $r$ is $0$ or $(r-1)$  had a legitimate end before time $t$. 
\end{definition}

\begin{definition} An acceptor trusts round $r$ if its $Safe\_r \geq r$.
\end{definition}

\begin{Lemma} \label{trustlegits}
If $r$ is a legit round, then eventually any correct acceptor will trust round $r$. 
\end{Lemma}
\begin{proof}
The proof is by induction on $r$.
\begin{itemize}
\item {\bf Base case}: for round $r=0$ each acceptor has $Safe\_r$ initialized to $0$. 
\item {\bf Inductive case}: By inductive hypothesis we have that eventually any acceptor sets  $Safe\_r=r-1$. Moreover, we have that, by definition of legit, round $(r-1)$  had a legitimated end. 
This means that $2f+1$ acceptors reliably broadcast $<ack, Accepted\_set, destination, \cdot,ts,r-1>$, and at least one of them is correct process $p$, thus the message $<ack, Accepted\_set, destination, \cdot,ts,r-1>$ is safe for $p$ at round $(r-1)$ . 
Any acceptor with  $Safe\_r=r-1$ upon receipt of these messages will eventually process them, by Lemma~\ref{gwtLemmasafe}, and it will set  $Safe\_r=r$, see procedure starting at line~\ref{gwts:accround}.
\end{itemize}\end{proof}

Note that by Lemma~\ref{trustlegits} we have that any legit round will be eventually trusted by all acceptors. Moreover, we can show that if $r$ is a non-legit round at time $t$ than it will not be trusted, at time $t$, by any correct
acceptor.
\begin{Lemma} \label{trustnonolegits}
If $r$ is a non-legit round, at time $t$, that is, $r\neq 0$ and $(r-1)$  has not had a legitimate end before time $t$, then any correct acceptor has $Safe\_r < r$.
\end{Lemma}
\begin{proof}
The proof derives immediately from the definition of legitimate end and from line~\ref{gwts:accround} in the acceptor code.
\end{proof}

\begin{definition} A value $v$ has been disseminated, by time $t$, if, by time $t$, it was contained in a safe $ack\_req$ message for some round $r$ and it has been received by $\lfloor (n+f)/2\rfloor +1$ correct acceptors that trusted round $r$. 
\end{definition}

The observation below is a strengthen version of Lemma~\ref{Lemmacomm}.
\begin{observation}\label{obs:inclusion}
If a value $v$ has been disseminated by time $t$, then any proposal committed after time $t$ will contain $v$.
\end{observation}
\begin{proof}
The proof is immediate by observing that a disseminated value is in the $Accepted\_set$ of $\lfloor (n-f)/2\rfloor +1$ correct acceptors (either by line~\ref{gwtaa} or~\ref{gwtab} of Algorithm 4), and by using the same argument of Lemma~\ref{Lemmacomm}. 
\end{proof}

\begin{Lemma} \label{correctdecision}
If round $r$ has a legitimate end and at least $(n-f)$ correct proposers joined round $r$, then eventually any correct proposer, that joined round $r$, will decide in round $r$, and join round $(r+1)$ .
\end{Lemma}
\begin{proof} The proof is by contradiction. 
Let $p_i$ be a correct process that joined round $r$ but has not yet decided in round $r$. Note that $p_i$ has to be in state $proposing$: by hypothesis $(n-f)$ correct proposers joined $r$, thus the guard at line \ref{gwt:label} has to be eventually triggered. 

Since $r$ has a legitimate end then there are $\lfloor (n+f)/2\rfloor +1$ reliable broadcast of messages  $<ack, Accepted\_set, destination, \cdot,ts,r>$,  and at least one of them has been generated by a correct process $p$, thus it is safe for $p$ at round $r$. 
By Observation~\ref{obs:inclusion} we have that $Decide\_set \subseteq Accepted\_set$, and, by Lemma~\ref{gwtLemmasafe}, upon receipt of these messages $p_i$ decides and joins round $(r+1)$ .
\end{proof}

\begin{Lemma} \label{eventualjoin}
If $r$ is a legit round, then any correct proposer eventually joins it. 
\end{Lemma}
\begin{proof}
The proof is by induction on round number. 
	\begin{itemize}
	\item {\bf Base case}: round $r=0$, by assumption it is a legit round, and by algorithm construction each correct proposer joins $r=0$. 
	\item {\bf Inductive case}: The inductive hypothesis is that $(r-1)$  is a legit round and that each correct proposer joined it.
	We assume that $r$ is a legit round, thus round $(r-1)$  had a legitimate end. Lemma~\ref{correctdecision} and the inductive hypothesis imply that any correct proposer joins $r$. 
	\end{itemize} \end{proof}

\begin{Lemma} \label{legitimateend}
If $r$ is a legit round, then it will eventually have a legitimate end. Moreover, each correct proposer executes line \ref{gwtsrefined}, while its round variable is $r$, at most $f$ times (that is it refines its proposal at most $f$ times during its participation to round $r$). 
\end{Lemma}
\begin{proof}
First observe, by Lemma~\ref{trustlegits} that each correct acceptor eventually trusts round $r$.  
Then observe that, until $r$ does not have a legitimate end, by Lemma~\ref{trustnonolegits}, no correct acceptor will trust any round $r' > r$.  Thus they will not process any message coming from round $r'$. 

The above and Observation~\ref{finitevalues} bound the number of changes that correct acceptors perform in round $r$ on their $Accepted\_set$ to a finite number. 
Therefore, there exists a time $t$ after which each correct acceptor does not change anymore its $Accepted\_set$.

If a correct proposer, that joined round $r$, issues an $ack\_req$ after time $t$, then, by Lemma~\ref{gwtLemmasafe}, and the above reasoning we have that such request will be committed. Once committed
round $r$ has a legitimate end.
Now by Lemma~\ref{eventualjoin} we have that eventually any correct proposer joins round $r$. 
It remains to show that some correct process issues a request after time $t$. Note that, when joining a new round, each correct process proposes its value. 
This proposal either is committed or refined (upon execution of line~\ref{gwtsrefined}). In case of refinement a new set is immediately proposed. This ensures that either something was committed in round $r$ before time $t$, or that
something will be proposed after $t$. In both cases round $r$ will have a legitimate end.
Recall that all correct proposers eventually join round $r$ (by Lemma \ref{eventualjoin}), each of them only proposes the value constituted by $Batch[r]$, by the property of the Byzantine reliable broadcast and by the safety of messages also Byzantines are constrained to propose
at most $f$ different values in  round $r$. This means that a decision in round $r$ can contain at most $n$ new values with respect to the decision at round $(r-1)$ . 
Since when a correct proposer executes line \ref{gwt:proposefirst} it passed the guard at line \ref{gwt:label}, it is obvious that there at most $f$ values missing in its proposal. From this Lemma 10, and thus the bound on the number of executions of line \ref{gwtsrefined}, follows immediately. \end{proof}

\begin{Lemma}\label{proposenolock}
If a correct process $p_i$  joins a round $r$ at time $t$, it also proposes, in round $r$, all values in $Pro_i$ received before time $t$. 
\end{Lemma}
\begin{proof}
Observe that a correct process joins a round $r$ only if $(r-1)$  had a legitimate end. From Lemma~\ref{eventualjoin} we have that all correct processes will eventually join round $r$, thus the if at line~\ref{gwt:label} will be passed.
From the above, the atomicity of the local procedures, and the fact that a correct process cumulates in its $Proposed\_set$ all previous batches never removing any value (see line~\ref{gwts:cumulation}), our claim follows. 
\end{proof}

From Lemma~\ref{proposenolock} we have the following observation:
\begin{observation} \label{dissemination}
Given any correct process $p_i \in C$ and any value  $v \in Pro_i$, we have that $v$ is eventually disseminated. 
\end{observation}

\begin{theorem}\label{th:GWTSliveness-inclusivity}
Let us consider a set of processes, of size $(3f+1)$, executing GWTS algorithm.  We have that any run of GWTS ensures the following liveness properties:
\begin{enumerate}
\item {Liveness:} each correct process $p_i \in C$ performs an infinite sequence of decisions $Dec_{i}=\{dec_{0},dec_{1}, dec_{2},\ldots\}$
\item {Inclusivity:} Given any correct process $p_i \in C$, if $Pro_i$ contains a value $v$ then, $v$ is eventually included in one decision in $Dec_{i}$;
\end{enumerate}
\end{theorem}
\begin{proof}
We prove each property separately: 
\begin{enumerate}
\item {Liveness:}  it is enough to show that there is an infinite sequence of legit rounds. This derives from Lemma~\ref{legitimateend}, combining it with Lemma~\ref{correctdecision} and a simple induction on the round number. 
Lemmas~\ref{correctdecision},~\ref{eventualjoin} ensure that in each round of such sequence all correct proposers decide.
\item {Inclusivity:} it derives from Observations~\ref{dissemination},~\ref{obs:inclusion} and the fact that the sequence of decisions is infinite (see above). \end{enumerate}\end{proof}

\subsection{Message Complexity}\label{sec:gwtmessagecomplexity}
GWTS executes a possibly infinite sequence of decisions. Thus, we restrict our message complexity analysis to the number of messages needed for each decision. The messages are counted per proposer, we include messages created by correct acceptors
in response to proposer actions. 
Each proposer decides exactly once for each algorithm round. Therefore, we count messages from start to end of a generic round.
A proposer has to reliably broadcast its batch (line \ref{gwt:rbcast1} -cost ${\cal O}(n^2)$), it has to broadcast its proposal (line \ref{gwt:proposefirst} - cost ${\cal O}(n)$), then, in the worst case, it refines its proposal at most $f$ times (see Lemma \ref{legitimateend} -line \ref{gwt:refine} - cost ${\cal O}(n)$), however each ack from a correct acceptor has to be reliably 
broadcast (line \ref{rback} - cost ${\cal O}(n^2)$). The total cost is therefore upper-bounded by ${\cal O}(f \cdot n^2)$.

\section{Byzantine tolerant RSM} % (fold)
\label{sec:byzantine_rsm}

We are interested in wait-free implementations of linearizable  replicated state machines for commutative update operations in the Byzantine model.

\subsection{Specification  of the Byzantine tolerant  RSM}

The replicated state machine is composed of $n$ replicas, which start in the initial empty state $s_0$. Among them, up to $f \leq (n-1)/3$ replicas may exhibit Byzantine failures. The RSM exposes two operations, update and read, such that the update operation with command  $cmd$ modifies the current state $s$ of the RSM by applying $cmd$ to $s$ but does not return any value, while the read operation returns the current state of the RSM. The state of the RSM at time $t$  is a set of update commands applied to the initial state $s_0$ until time $t$. Note that being the RSM commutative the order in which updates are applied does not matter. Clients may trigger an infinite number of read and update operations. We assume that each command is unique (which can be easily done by tagging it with the identifier of the client and a sequence number). We do not make any assumptions regarding clients behavior: they can exhibit arbitrary behaviors (e.g., invoke an update operation with some arbitrary command, or modify the read and update code). We do not limit the number of Byzantine clients. Hence, to prevent Byzantine clients from jeopardizing the state of the RSM through  arbitrary commands, commands are locally executed by clients: the RSM provides clients with a set of updates and clients locally execute them. For readibility reasons, the value returned by the execution of a set of commands is equal to the set of commands.
The following properties formalise the behavior of read and update operations during any execution run by correct clients:

\begin{itemize}
\item \textbf{Liveness} Any update and read operation completes;
\item \textbf{Read Validity: }  Any value returned by a read  reflects a state of the RSM;
\item \textbf{Read Consistency: }  Any two  values returned by any two reads are comparable;
\item \textbf{Read Monotonicity: } For any two reads $r_1$ and $r_2$ returning value $v_1$ and $v_2$ respectively, if $r_1$ completes before $r_2$ is triggered then $v_1 \subseteq v_2$; 
\item \textbf{Update Stability: } If update $u_1$ completes before  update $u_2$ is triggered then every read that returns a value that includes the command of $u_2$ also include the command of $u_1$;
\item \textbf{Update Visibility: } If update $u$ completes before read $r$ is triggered then the  value returned by $r$ includes the command  of $u$.
\end{itemize}

\subsection{Implementation of the Byzantine tolerant  RSM}
As previously introduced, our general idea to implement a wait-free and linearizable replicated state machine resilient to Byzantine failures in an asynchronous system is to apply  Generalized Lattice Agreement on the power set of all the update commands. GWTS is  executed by the replicas of the state machine (for simplicity reasons  replicas play the role of both proposers and acceptors). The update and read operations are presented in Algorithms~\ref{alg:rsmupdate} and~\ref{alg:rsmread} respectively.  The update operation consists in submitting the new command $cmd$ to generalized Lattice Agreement so that eventually the new state of each (correct) replica includes $cmd$. This is achieved by triggering the  execution of  \textsc{new value} with $\{cmd\}$ as parameter at any subset of $(f+1)$ replicas (so that at least one correct replica will execute it), see Line~\ref{newvalue}. The update operation  completes when some correct replica modifies its local state with $cmd$, that is, decides a decision value that includes $cmd$  (Line~\ref{writeupdate}). This preserves the order of non-overlapping update operations.
The read operation consists in an update operation followed by a confirmation step. The update is triggered with a special value $nop$ that locally modifies a replica's state as for an ordinary command $cmd$ but is equivalent to a nop operation when executed. When the update completes, any decision value decided by a correct replica can be returned by the read operation. Since up to $f$ Byzantine replicas may provide any value, a confirmation request for each of these $(f+1)$ decision values  is sent to all replicas (Line~\ref{rsm:askconfirmation}).  A replica acknowledges $Accepted\_set$ if $Accepted\_set$ has  been accepted by $\lfloor (n+f)/2\rfloor+1$  acceptors, which ensures that $Accepted\_set$ has effectively been decided in GWTS (Line~\ref{gtwsplug:sendack} of Alg. \ref{GWTSplugin}).  The value returned by the read operation is the result of the execution of the first decision value confirmed  by $f+1$ replicas, i.e, the first decision value confirmed by at least one correct replica. This ensures that a read operation will return a value that reflects the effect of the last update operation.
From an implementation point of view, the confirmation step requires to add two lines of code in Algorithm~\ref{gwts:prop}. Specifically, when a proposer receives a confirmation request for decision value $Accepted\_set$, then it acknowledges the request if $<ack, Accepted\_set,\cdot,\cdot,ts,r>$ appears $\lfloor (n+f)/2\rfloor+1$ times in its $Ack\_history$ set for a fixed combination of $ts$ and round $r$. See Code in Alg.~\ref{GWTSplugin}. 
\begin{algorithm*}
\caption{\emph{Replicated State Machine - Update algorithm at  a client}}
\label{alg:rsmupdate}
\footnotesize

\begin{algorithmic}[1]

 \Procedure{Update}{$cmd$}  
 \State $DecSet=\emptyset$
 \State {\textsc{new value} $(\{cmd\})$ at  $(f+1)$ \sc{replicas}}{}
 \label{newvalue}
\State {\textbf{wait until} $|DecSet| \geq f+1$}{} \label{writeupdate}

 \smallskip
\algrenewcommand\algorithmicprocedure{\textbf{upon event}}
\Procedure{Receipt from $replica <~\textsc{decide},Accepted\_set,replica>$ with $cmd \in Accepted\_set$}{} 
\State $DecSet = DecSet \; \cup <\textsc{decide},Accepted\_set,replica>$
\EndProcedure
\EndProcedure
\end{algorithmic}
\end{algorithm*}

%%%%%%%%%%%%%%%%%%%%
\begin{algorithm*}
\caption{{\em Replicated State Machine - Read algorithm at  client $c$}}
\label{alg:rsmread}
\footnotesize

\begin{algorithmic}[1]
%\Comment{$r$-th read request. Initially set to $0$}
%\State $Dec\_Set=\emptyset$  
%\Comment{set containing the first $(f+1)$  decision values}
%\Comment{set of confirmation of a decision values. }
%\smallskip
%\smallskip
\Procedure{Read}{}
\State{$DecSet=ConfSet =  \emptyset$ }
\State {\textsc{new value} $(\{nop_{c,r}\})$ at  $(f+1)$ \sc{replicas}}{}\label{deb-update}

\smallskip
\algrenewcommand\algorithmicprocedure{\textbf{upon event}}
\Procedure{Receipt from $replica$ $<\textsc{decide},Accepted\_set,replica>$ with $nop_{c,r} \in Accepted\_set$}{} 
\State $DecSet = DecSet \; \cup <\textsc{decide},Accepted\_set,replica>$
\EndProcedure
\smallskip
\algloop{foreach}
\Procedure{$|DecSet| \geq f+1$}{}\label{end-update}
\ForAll{ $Accepted\_set \in DecSet$}\label{start-confirmationphase}
\State {\sc Send}$(<\textsc{CnfReq},Accepted\_set>)$ to all replicas \label{rsm:askconfirmation}
\EndFor
\EndProcedure
\smallskip
\Procedure{Receipt from $replica$ $<\textsc{CnfRep},Accepted\_set,replica>$}{}
\State $ConfSet = ConfSet \; \cup <\textsc{CnfRep}, Accepted\_set,replica>$
\EndProcedure
\smallskip
\Procedure{$<\textsc{CnfRep}, Accepted\_set,.>$ appears $(f+1)$ times in $ConfSet$}{}\label{end-confirmationphase}
\State {\textbf{return} \sc{execute} ($Accepted\_set$)}{}
\EndProcedure
\EndProcedure

\end{algorithmic}
\end{algorithm*}

\begin{algorithm*}
\caption{GWTS - proposer $p_i$ \emph{code plug-in for supporting RSM}}\label{GWTSplugin}
\footnotesize
\begin{algorithmic}[1]
\State $Pending\_conf=\emptyset$ 
\smallskip
\algrenewcommand\algorithmicprocedure{\textbf{upon event}}
 \Procedure{Deliver confirmation\_req from Client $c$}{$<\textsc{Conf\_req},Accepted\_set>$}  
\State $Pending\_conf=Pending\_conf \cup \{<\textsc{Conf\_req},Accepted\_set,c>\}$
\EndProcedure
 \smallskip
\Procedure{$\exists <\textsc{Conf\_req},Accepted\_set,c> \in Pending\_conf:$ $<\cdot,Accepted\_set,\cdot,\cdot,timestamp,r>$ appears $\lfloor (n+f)/2 \rfloor+1$ times in $Ack\_history$}{}  \label{gtwsplug:sendack}
\State {\sc Send to Client $c$}$(<\textsc{Conf\_rep},Accepted\_set,p_i>)$
\State $Pending\_conf=Pending\_conf \setminus \{<\textsc{Conf\_req}, Accepted\_set,c>\}$
\EndProcedure
\end{algorithmic}
\end{algorithm*}

\begin{theorem}
Given the wait-free Byzantine generalized Lattice Agreement (GLA) algorithm whose pseudocode is given in  Alg.~\ref{gwts:prop},~\ref{gtws:acceptor} and~\ref {GWTSplugin},
the above transformation yields a wait-free linearizable replicated state machine for commutative update operations. This transformation requires one execution of the Byzantine GLA algorithm. 
\end{theorem}
\begin{proof}
The proof consists in showing that (1) liveness, (2) read validity, (3) read consistency, (4)  update stability, (5) read monotonicity, and (6) update visibility properties holds. We prove each property separately.
\begin{itemize}
\item (1) Liveness of the update operation is straightforward from Theorem~\ref{th:GWTSliveness-inclusivity}. For the read operation, liveness holds from update liveness and from the fact that among the $(f+1)$ received values, at least one is the decision value of a correct replica, which by Lines~\ref{livenessRSM-b}-~\ref{livenessRSM-e} of Algorithm~\ref{gwts:prop}, has been accepted and reliably broadcast to all proposers by $(2f+1)$ acceptors, and thus reliably delivered by all proposers (by Liveness of Reliable Broadcast);
\item (2) Straightforward from the fact that the value returned by a read is a decision value.
\item (3) Straightforward from Theorem~\ref{th:gwtssafety};
\item (4) By Observation~\ref{obs:inclusion} and by the fact that a read operation begins with an update operation;
\item (5) By applying the same argument as for update stability, read monotonicity holds.
\item (6) By applying the same argument as for update stability, update visibility holds.
\end{itemize}
\end{proof}

\begin{Lemma}
The above transformation is resilient to Byzantine clients. 
\end{Lemma}
\begin{proof}
A Byzantine client $c$ may threaten the correctness of the above transformation by uncorrectly executing the read and write operations. 
\begin{itemize}
\item Suppose that $c$ invokes Update($cmd$). If $cmd$ is not a admissible command (i.e., is not an element of the lattice) then correct replicas filter out $cmd$ (line~\ref{gwt:lattice}), and thus the operation has no impact on the RSM state. Now, if $c$ sends $cmd$ to less than $f+1$ replicas (line~\ref{newvalue}), then it is sufficient that a single replica receives and proposes $cmd$ to GWTS for $cmd$ to be decided by all correct replicas, and thus correctly impacts the state of the RSM. Finally if $c$ invokes a sequence of updates $u_i, \ldots u\ell$ without waiting for each $u_i$ to complete (line~\ref{writeupdate}), then these update operations will be considered as concurrent operations, which is  handled by GWTS. 
\item Suppose that $c$ invokes Read(). The only difference with the Update() operation is the confirmation phase (lines~\ref{start-confirmationphase}-~\ref{end-confirmationphase}) which acts as a proof that the value returned by the read operation is a decided value. This phase does not impact the state of the RSM, which completes the proof of the lemma.
\end{itemize}
\end{proof}

\section{Safety by Signature Algorithm - An algorithm with linear message complexity}
\label{sec:sign}

In this Section we will discuss how to decrease the message complexity using signatures. The resulting algorithm has a message complexity of ${\cal O} (n)$ when $f={\cal O}(1)$. This Algorithm introduces a trade-off between message-complexity and message-size:  it sends messages that could have size
${\cal O}(n^2)$, this is not the case with the original WTS.
Pseudocode is on Algorithms \ref{sbs:one},\ref{sbs:two}, and \ref{sbs:three}.
The algorithm is conceptually divided in three phases:
\begin{itemize}
\item {\em Init}: in this phase each process broadcasts a signed version of its initial proposed value. A process collects these messages until it sees $(n-f)$ of it. The purpose of this phase is to ensure
that any correct process is able to create a set of values containing at least $|X \cup B|-2f$ values from correct processes. 
\item {\em Safetying}: At the end of the init phase a correct process has a certain set of values. The purpose of the safetying phase is to make at least $|X \cup B|-2f$ such values safe. In such case a value $v$ is safe if we are sure that no other process can see a different value 
$v'$ that is also safe and has been sent by the same sender of $v$. Safetying is done by performing a broadcast of the $(n-f)$ signed values obtained in the init phase towards the acceptors. 
Each acceptor keeps a set of values that are candidate to be safe. Once an acceptor receives a set of values from a proposer, it starts processing each value contained in it.
If there is a value $v$ such that it has not see another value sent by the same sender (this check is done by using signature), it adds the value $v$ to its candidate set.
Otherwise, if it exists a $v'$ from the same sender, it adds $(v,v')$ to a temporary set of Conflicts. 
The acceptor replies back to the proposer by sending a signed message that contains the set received, and the set of conflicts found. 
A proposer possesses a proof of safety for a value $v$ if it receives $\lfloor (n+f)/2 \rfloor+1$ messages from different acceptors in which $v$ never appears as a conflict. 
The intuition behind this phase is that if the same Byzantine process injects two (or more) values signed by him in the init phase, then at most one of them could manage to get a correct proof of safety. 
\item {\em Proposing}: This phase is identical to the proposing phase of WTS. The only difference is that a correct proposer and acceptor refuse to process any message that contains a value without an
attached proof of safety. 
\end{itemize}

\begin{algorithm*}
\caption{{\sc SbS} One - Algorithm for proposer process $p_i$}\label{sbs:one}
\footnotesize

\begin{algorithmic}[1]
\State $V$ lattice values, $E$ subset of $V$ of values that can be proposed by processes. 
\State $Proposed\_set= Safe\_acks=Ack\_set=\emptyset$
\State $byz[\forall k \in [0,n]]=False$ \Comment{Each location $i$ stores a boolean indicating if $acceptor_i$ is known to be Byzantine. }
\State $ts=0$
\State $Safety\_set=\emptyset$
\State $proposed\_value=\bot$ 
\State $state = init$
\medskip

\algrenewcommand\algorithmicprocedure{\textbf{upon event}}

\LineComment{Init Phase}
\Procedure{$proposed\_value \neq \bot$}{}
\State  $payload=<${\sc Sign}$(proposed\_value),p_i>$ \Comment{signed value}
\State  $Safety\_set=Safety\_set \cup \{payload\}$ \label{sbs:insertvalue}
\State  {\sc BrodcastToProposers}$(<init\_phase, payload>)$ \label{sbs:firstbroadcast}
\EndProcedure
\medskip

\Procedure{BcastDelivery from sender}{$<init\_phase, <value,sender>>$}
\If{{\sc Verify}$(value) \land value\in E\land state=init$}
\State  $Safety\_set=\text{ {\sc RemoveConflicts}}(Safety\_set \cup \{<value,sender>\})$

\EndIf
\EndProcedure
\medskip

\LineComment{Safetying phase}
\Procedure{upon event $|Safety\_set| \geq n-f \land state=init$}{}\label{start:safetying}
\State $state=safetying$
\State  {\sc BrodcastToAcceptors}$(<safe\_req, Safety\_set>)$ \label{safereq:send}
\EndProcedure
\medskip

\Procedure{Delivery {\bf when} $state=safetying$ from sender}{$m=<safe\_ack, Rcvd\_set, Conflicts>$}

\If{{\sc Verify}$(m) \land Rcvd\_set=Safety\_set \land \forall (x,y) \in Conflicts |$ {\sc VerifyConfPair}$((x,y))$} \label{safeack:check}
\State  $Safe\_acks=Safe\_acks \cup \{(<safe\_ack, Rcvd\_set, Conflicts>,sender)\}$
\Else
\State $byz[sender]=True$
\EndIf

\EndProcedure
\medskip

\LineComment{Proposing Phase}
\Procedure{$|Safe\_acks| \geq \lfloor (n+f)/2 \rfloor +1 \land state=safetying$}{} \label{procedure:safe}

\ForAll{$<v,s> \in Safety\_set$ such that $(\nexists (<\cdot,\cdot,Conflicts,\cdot>,\cdot) \in Safe\_acks | (<v,s>,\cdot) \in Conflicts \lor (\cdot,<v,s>) \in Conflicts)$}
\State $Proposed\_set=Proposed\_set \cup \{<v,Safe\_acks>\}$ \label{sbs:insertsafe}
\EndFor
\State $state=proposing$ 
\State $Ack\_set=\emptyset$
\State $ts=ts+1$
\State  {\sc BrodcastToAcceptors}$(<ack\_req, Proposed\_set,ts>)$ \label{sbs:proposalbroadcast}
\EndProcedure
\medskip

\Procedure{Delivery {\bf when} $state=proposing$ from sender}{$<ack,Rcvd\_set,rts>$}
\If{$rts=ts$}
\If{$Proposed\_set = Rcvd\_set \land byz[sender]=False$}
\State  $Ack\_set=Ack\_set \cup \{(<ack,Proposed\_set,rts>,sender)\}$
\Else
\State $byz[sender]=True$
\EndIf
\EndIf
\EndProcedure
\medskip

\Procedure{Delivery  {\bf when} $state=proposing$ from sender}{$<nack,Rcvd\_set,rts>$}
\If{$ rts=ts$}
\If{$Proposed\_set \neq Rcvd\_set \cup Proposed\_set \land byz[sender]=False \land ${\sc AllSafe}$(Rcvd\_set)$} \label{sbs:refguard}
\State $Proposed\_set=Rcvd\_set \cup Proposed\_set$ \label{sbs:refine1}
\State $Ack\_set= \emptyset$
\State $ts=ts+1$
\State  {\sc BrodcastToAcceptors}$(<ack\_req, Proposed\_set,ts>)$ \label{sbs:refine}
\Else
\State byz[sender]=True
\EndIf
\EndIf
\EndProcedure
\medskip

\Procedure{$|Ack\_set| \geq \lfloor (n+f)/2 \rfloor +1 \land state=proposing$}{}
\State $state=decided$ \label{sbs:decideonce}
\State $Only\_values=\{v : \exists <v,\cdot> \in Proposed\_set\}$
\State {\sc DECIDE}($Only\_values$)
\EndProcedure

\end{algorithmic}
\end{algorithm*}

\begin{algorithm*}
\caption{{\sc SbS} One - Algorithm for acceptor process $p_i$}\label{sbs:two}
\footnotesize

\begin{algorithmic}[1]
\algrenewcommand\algorithmicprocedure{\textbf{upon event}}
\State $SafeCandidates=\emptyset$
\State $Accepted\_set = \emptyset$

\medskip
\Procedure{BcastDelivery from sender}{$<safe\_req,Safety\_set>$} \label{asafereq}
\If{$\forall <v,s> \in Safety\_set |${\sc Verify}$(v)$}
\State {\sc Send to sender}({\sc Sign}($<ack, Safety\_set,$ {\sc ReturnConflicts}$(Safety\_set \cup SafeCandidates),x>$)) \label{aconf:1}
\State $SafeCandidates=SafeCandidates \, \cup \, ${\sc RemoveConflicts}$(Safety\_set \cup SafeCandidates )$ \label{aconf:2}
\EndIf
\EndProcedure
\medskip
\Procedure{BcastDelivery from sender}{$<ack\_req,Rcvd\_set, x>$}
\If{{\sc AllSafe}$(Rcvd\_set)$} \label{allsafe:accp}
\If{ $Accepted\_set \leq
 Rcvd\_set$}
 \State $Accepted\_set=Rcvd\_set$
\State {\sc Send to sender}$(<ack, Accepted\_set,x>)$
\Else
\State  {\sc Send to sender}$(<nack, Accepted\_set,x>)$
\State $Accepted\_set=Rcvd\_set \cup Accepted\_set$
\EndIf
\EndIf
\EndProcedure

\end{algorithmic}
\end{algorithm*}

\begin{algorithm*}
\caption{Helper Procedures}\label{sbs:three}
\footnotesize

\begin{algorithmic}[1]

\Procedure{ReturnConflicts}{$Set$}
\State $Result=\emptyset$
\ForAll{$(x,y) \in Set \times Set |$ s.t. {\sc VerifyConfPair}$((x,y))$}
\State  $Result=Result \cup \{(x,y)\}$
\EndFor
\State return $Result$
\EndProcedure
\medskip

\Procedure{RemoveConflicts}{$Set$}
\State $Result=Set$
\ForAll{$(x,y) \in ${\sc ReturnConflicts}$(Set)$}
\State $Result=Result \setminus \{x,y\}$
\EndFor
\State return $Result$
\EndProcedure
\medskip

\Procedure{VerifyConfPair}{$(x,y)$}
\State return {\sc Verify}($x$) $\land$ {\sc Verify}($y$) $\land$ $x$.sender = $y$.sender $\land x.value \neq y.value$
\EndProcedure

\Procedure{AllSafe}{$Set$}
\State $Result=True$
\ForAll{$<v,Acks> \in Set$}
\If{$|Acks| \geq \lfloor (n+f)/2 \rfloor +1 \land( \nexists <\cdot,Proposal,\cdot,\cdot> \in Acks | v \not\in Proposal) \land( \nexists <\cdot,\cdot,Conflicts,\cdot> \in Acks | (v,\cdot) \in Conflicts \lor (\cdot,v) \in Conflicts) \land (\nexists m \in Acks | \neg${\sc Verify}$(m)$ is not valid $) \land (\nexists m,m' \in Acks | m \neq m' \land m.$sender$=m'.$sender$) \land v \in E$}
\State $Result=Result \land True$
\Else
\State $Result=Result \land False$
\EndIf
\EndFor
\State return $Result$
\EndProcedure
\medskip

\Procedure{Sign}{$e$}
\State Signs the element $e$ (that could be a message or a value) and it returns a new element $e'$ that is a signed version of $e$.
\EndProcedure

\Procedure{Verify}{$e$}
\State Takes an element $e$, and it returns true if an only if $e$ has a correct signature: $e$ has been returned by a call to the above Sign procedure.
\EndProcedure

\end{algorithmic}
\end{algorithm*}

\subsection{Algorithm properties}

We say that a value $v$ and a value $v'$ conflict if $v \neq v'$ and $v.sender=v'.sender$. Given a set of values we say that the set is conflicting if it contains two conflicting values.

\paragraph{Safety.} 
In the signature based algorithm we define a value as safe if:
\begin{definition} (Safe value) A value $v$ is safe if there are $\lfloor (n+f)/2 \rfloor+1$ broadcasts of messages  $<safe\_ack, Rcvd\_set, Conflicts, rts>$ from different acceptors, such that
$v$ is contained in each $Rcvd\_set$ and it is not contained in a pair inside a set $Conflicts$.
\end{definition}

\begin{Lemma}\label{sbs:safeensure}
Let $W_j$ be the set of values signed by process $p_j \in P$ appearing during any execution of the SbS algorithm.
$W_j$ contains at most one safe value. 
\end{Lemma}
\begin{proof}
By contradiction let us assume that there are two values $x,y \in W_j$ both safe. 
Then for each value we have  $\lfloor (n+f)/2 \rfloor +1$ acceptors that broadcast the $safe\_ack$ messages. 
Any two sets of $\lfloor (n+f)/2 \rfloor +1$ acceptors have a non empty intersection in the set of correct acceptors: there should be at least one correct acceptor $p$ that generated the message (or messages) $safe\_ack$  that made both  $x$ and $y$ safe.
First of all, notice that the message cannot be the same for both values (since they conflict they would be inserted in $Conflicts$ by $p$, see line~\ref{aconf:1} of the acceptor code).
For this reason we should have that $p$ has done two different broadcasts one for $x$ and one for $y$. 
Let us assume w.l.o.g that the first broadcast has been for $x$, note that $x$ is inserted in $SafeCandidate$ set (line~\ref{aconf:2}).
This implies that once a $safe\_req$ with set containing $y$ is received by $p$ (line~\ref{asafereq}) $p$ detects the conflict with $x$ in $SafeCandidate$ (line~\ref{aconf:1}). 
Therefore, $p$ cannot generate any message in which $y$ does not appear in $Conflicts$. 
\end{proof}

\begin{Lemma}\label{sbs:safecorrect}
Let us consider an execution of SbS.
Let $p_i$ be a correct process, and let $v_i$ be a value that $p_i$ signs and inserts in its $Proposed\_set$ at line~\ref{sbs:insertvalue}.
If $p_i$ decides a set $decide_i$, then $v_i \in decide_i$. 
\end{Lemma}
\begin{proof}
First of all, note that $p_i$ is the only one that can create a valid signature on a value $v$. Therefore, if a Byzantine process creates a pair in a $Conflicts$ set containing value $v_i$ for a certain message  $m=<safe\_ack, \cdot, Conflicts, \cdot>$ , then no correct process will insert $m$ in its $Safe\_ack$ set (the signature check at line~\ref{safeack:check} will not be passed). Stated otherwise, $v_i$ can never appear in a $Conflicts$ set of a message $m=<safe\_ack, \cdot, Conflicts, \cdot>$  that will be processed by any correct process.
This implies that, if $p_i$ decides, then $p_i$ inserted $v_i$ in $Safe\_proposal$ at line~\ref{sbs:insertsafe}. 
Since, $p_i$ never removes any value from $Proposed\_set$ we have that an eventual decision contains $v_i$.
\end{proof}

Note that the proposing phase of SbS is analogous to the one of the WTS algorithm. It is not hard to see that  Lemma~\ref{Lemmacomm} and Observation~\ref{obsobv} holds also for the proposing phase of SbS.

\begin{theorem}\label{SuSA}
Let us consider a set of processes, of size at least $(3f+1)$, executing SbS algorithm. 
Algorithm SbS enforces: (1) {Comparability}; (2) {Inclusivity}; (3) {Non-Triviality}; (4) {Stability}. 
\end{theorem}
\begin{proof}
We prove each property separately. 
\begin{enumerate}
 \item is implied by Lemma~\ref{Lemmacomm}, Observation~\ref{obsobv} and the fact that a proposer refuses to process nacks containing values that are not safe (see if at line~\ref{sbs:refguard}).
 \item by Lemma~\ref{sbs:safecorrect}. 
\item the bound on $B$ and $X$ derives from the safety of each value inserted in a decision and Lemma~\ref{sbs:safeensure}.
\item  is ensured by line~\ref{sbs:decideonce} in the code run by proposers.  
\end{enumerate}\end{proof}

\paragraph{Liveness} The main result of this section is a bound of ${\cal O}(f)$ on the number of times a proposer refines its proposal. This bound will be needed to prove that our algorithm uses a quadratic number of messages.  

\begin{Lemma}\label{lemmaDelays}
Let us consider an execution of SbS. Any correct process eventually executes line~\ref{sbs:proposalbroadcast}, and when it does its $Proposed\_set$ contains at least $|X \cup B|-2f$ values.  
\end{Lemma}
\begin{proof}
Since at most $f$ processes are faulty line~\ref{start:safetying} will be eventually executed on each proposer, and when it happens its $Safety\_set$ contains at least $|X \cup B|-f$ values.
It remains to show that line~\ref{sbs:proposalbroadcast} is eventually executed and that at most $f$ values are not inserted in $Proposed\_set$  during the loop at line~\ref{sbs:insertsafe}.
First of all notice that at least $(n-f)$ acceptors will reply to the $safe\_req$ sent at line~\ref{safereq:send}, this ensures that eventually the procedure at line~\ref{procedure:safe} will be executed. 
Moreover, we have that the $Safety\_set$ sent at line~\ref{safereq:send} contains at most $f$ values of Byzantine processes. By the same reasoning used in the proof of Lemma~\ref{sbs:safecorrect} we have that at least
 $|X \cup B|-2f$ values are from correct proposers and will be sent at line ~\ref{sbs:proposalbroadcast}.
\end{proof}

\begin{Lemma}\label{LemmaX}
A correct proposer refines its proposal (executing line~\ref{sbs:refine}) at most $2f$ times. 
\end{Lemma}
\begin{proof}
Each time a correct proposer executes line~\ref{sbs:refine}, it passes the if at line~\ref{sbs:refguard}, thus increasing its $Proposed\_set$. However, its first proposal, in line~\ref{firstproposal}, contains at least $|X \cup B|-2f$ values. Since
there are at most $|X \cup B|$ safe values (from Lemma~\ref{sbs:safeensure}) and the if at line~\ref{sbs:refguard} can be passed only when all values inside $Rcvd\_set$ are safe, the claim follows.
\end{proof}

\begin{Lemma}\label{LemmaTF}
When a correct acceptor receives a message $m=<ack\_req,Rcvd\_set, x>$ from a correct proposer it answers with an ack or a nack. 

\end{Lemma}
\begin{proof}
If a correct acceptor sends a message $m$ then each value in $Rcvd\_set$ is safe, thus, when $m$ is delivered by a correct acceptor, the if at line~\ref{allsafe:accp} will be passed and the acceptor will answer with an ack or a nack. 
\end{proof}

{\it Mutatis mutandis}, an analogous of Lemma~\ref{Lemmaeve} holds also for the proposing phase of SbS:

\begin{Lemma} \label{SbS:Lemmaeve}
If there is a time $t$ after which a correct proposer $p_i$ in state $proposing$  cannot execute line~\ref{sbs:refine}, then $p_i$ eventually decides. 
\end{Lemma}
\begin{proof}

Let $<ack\_req,Proposed\_set,ts>$ be the last ack request message sent by $p_i$. Since $p_i$ does not execute line~\ref{sbs:refine1} it means that either it does not receive any nack, or that any nack it receives 
does not allow him to pass the if line~\ref{sbs:refguard}. Since $p_i$ is correct its message  $<ack\_req,Proposed\_set,ts>$ will reach each correct acceptor. By hypothesis each of them will send a ack, otherwise $p_i$ should
be able to execute line~\ref{sbs:refine1} (they all handle the ack request by Lemma~\ref{LemmaTF}). Once $p_i$ receives the acks from the set of correct acceptors it will handle them, Lemma~\ref{LemmaTF}, and decides. 
\end{proof}

\begin{theorem}\label{liveness:sbs}
Let us consider a set of processes, of size at least $(3f+1)$, executing SbS algorithm. 
Every correct proposer decides in at most $5+4f$ message delays. 
\end{theorem}
\begin{proof}
The broadcast at line~\ref{sbs:firstbroadcast} takes at most $1$ message delay. Therefore after one message delay each correct process starts its safetying phase. This phase takes two message delays: one to reach the acceptors, and one to receive the response back. Each refinement takes at most $2$ message delays, the time needed to broadcast
and receive a response. There are at most $2f$ refinements, see Lemma~\ref{LemmaX}. Executed in at most $4f$ message delays, by Lemma~\ref{SbS:Lemmaeve} after $5+4f$ message delay a correct has to decide.
\end{proof}

\paragraph{Message Complexity}
We now bound the message complexity by counting the maximum number of messages needed by a correct proposer to decide. We do it phase by phase:
\begin{itemize}
\item Init phase: each correct proposer executes a single broadcast. Cost ${\cal O}(n)$. 
\item Safetying phase: each correct proposer executes a single broadcast, such broadcast triggers at most one point-to-point response from each correct acceptor. Cost ${\cal O}(n)$.
\item Proposing phase: each correct proposer executes a single broadcast for each proposal refinement and one broadcast for its initial proposal. Each of these broadcasts  triggers at most one point-to-point response from each correct acceptors. Cost ${\cal O}(f \cdot n)$.
\end{itemize}
Summing the above costs we have a bound of ${\cal O}(n)$ messages. 

\subsection{Generalising SbS}
Adapting the SbS algorithm to its generalised version, while keeping the message complexity improvement, needs a special attention to substitute the reliable broadcast used to acknowledge in the GWT
(line~\ref{rback} of Algorithm~\ref{gtws:acceptor}). We would like to replace such broadcast with a single point-to-point message. 

Conceptually, that reliable broadcast has two functions: (1) it implicitly disseminates to all proposers the fact that someone is able to decide in a certain round (sharing also a possible decision value),
(2) it forces a Byzantine process to broadcast the same ack message to all proposers, enforcing the ``publicity" of such action. 

These two functions can be replaced by two modifications. 
The second function could be replaced by forcing an acceptor to sign its, now poin-to-point, ack. Intuitively this allows a proposer to prove others that it 
received a certain ack for its proposal request. The first function is replaced by letting any correct proposer broadcast a special decided message before deciding, such message has attached 
all the acks used to decide. This would allow proposers and acceptors that receive a decided message to know that the sender of such message was allowed to decide by
the algorithm rules (recall that acks are now signed). Similarly to what happen in GWTS, a round $r$ ends when someone broadcasts a well-formed decided message for that round.  Where with well-formed we mean that 
such decided message has attached $\lfloor (n+f)/2 \rfloor+1$ signed acks. 
Therefore, a correct acceptor will trust a round $r$, only if it trusted round $(r-1)$  and it knows that rounds $(r-1)$  terminated (this knowledge derives from seeing a decided message for round $(r-1)$ ).
Finally, to avoid the reliable broadcast of decided messages, a correct acceptor piggybacks a decided message, that it has seen for a round $r$, in the replies to ack\_requests for round $r$. 
\paragraph{Message Complexity}
The message complexity follows the same lines of the analysis in Section \ref{sec:gwtmessagecomplexity}. However, since we are not using the Byzantine reliable broadcast, messages per decision on each proposer are at most ${\cal O}(f \cdot n)$.

\section{Conclusions} % (fold)
\label{sec:conclusions}
We investigated Byzantine Lattice Agreement and we used it to build a byzantine tolerant RSM with commutative updated. There are several directions meriting investigation. 
Our main interest is to understand whether a message delay of ${\cal O}{(f)}$ is necessary or not. In the crash-stop model a recent paper \cite{Garg:2018b} has shown that ${\cal O}{(\log f)}$ delays are sufficient.
Therefore, a first step would be to investigate if the technique used in \cite{Garg:2018b} could be ``Byzantined" while preserving the desirable delay. A final target is to understand the necessary number of message delays.
This would be extremely interesting, even in the crash-stop model such knowledge is still missing. 
% section conclusions (end)

\bibliographystyle{IEEEtran}
\bibliography{main} 

\end{document}